%% file: main.tex
\documentclass[aps,pra,twocolumn,showpacs,superscriptaddress,nofootinbib]{revtex4-2}

\usepackage{graphicx}  % needed for figures
\usepackage{tikz}
\usetikzlibrary{arrows.meta, decorations.markings,shapes.geometric}

\usepackage{float}
\usepackage[toc,page]{appendix}
\usepackage[table]{xcolor}
\usepackage{mathtools}
\usepackage{dcolumn}   % needed for some tables
\usepackage{bm}        % for math
\usepackage{color}

\usepackage{hyperref}
\usepackage{pifont}
\usepackage{mathrsfs}
\usepackage{amsmath, amsthm, amssymb}
\usepackage{amsfonts}
\usepackage{relsize}    % to resize the sum symbol etc
\usepackage{bbold}       % for identity in tensor notation
\usepackage{accents}    % \underaccent { accent } { symbol }

\hypersetup{
	colorlinks=true, %set true if you want colored links
	linktoc=all,     %set to all if you want both sections and subsections linked
	linkcolor=blue,  %choose some color if you want links to stand out
	citecolor=blue,
	filecolor=blue,
	urlcolor=blue
}
\usepackage{soul}
\newcommand{\beq}{\begin{equation}}
	\newcommand{\eeq}{\end{equation}}
\newcommand{\bqa}{\begin{eqnarray}}
	\newcommand{\eqa}{\end{eqnarray}}

\newcommand{\erf}[1]{Eq.~(\ref{#1})}

\newcommand{\erfa}[2]{Eqs.~(\ref{#1}) and (\ref{#2})}
 
\newcommand{\srf}[1]{Sec.~\ref{#1}}

\definecolor{BLACK}{gray}{0}
\definecolor{RED}{rgb}{1,0,0}
\definecolor{GREEN}{rgb}{0.2,.6,0.2}
\definecolor{amber}{rgb}{1.0,0.49,0.0}

\newcommand{\bra}[1]{\langle{#1}|}

\newcommand{\ket}[1]{|{#1}\rangle}

\newcommand{\op}[2]{\ket{#1}\bra{#2}}

\newtheorem{theorem}{Theorem}

\begin{document}
	
	\title{Robustness hierarchy of bipartite quantum correlations under noisy dynamics}
	
	\input author_list.tex

	\date{\today}
	
	\begin{abstract}
		Quantum correlations are fragile under noise, motivating quantitative measures that track not only
		their presence but also their degradation during physical evolution. We develop a robustness
		framework for the hierarchy of Bell nonlocality, EPR steering, and entanglement, establishing an
		ordering of the corresponding robustness measures and relating robustness to noise-induced
		transition times. We further formulate the instantaneous robustness decay along Lindblad dynamics
		as a directional derivative determined by both the evolution and the robustness functional.
		Applications to isotropic depolarizing noise and amplitude damping illustrate, respectively, the
		resulting resource transition times and the directional dependence of steering robustness. The
		framework provides an approach to characterizing the degradation and lifetime of bipartite quantum
		correlations under noisy dynamics.
	\end{abstract}

	\maketitle
	
	\section{Introduction}
	
	Quantum correlations are among the most distinctive features of quantum theory and constitute key
	resources for quantum information processing
	tasks~\cite{HorHor09,AciWil18,PirWal20,BhaAsp22,MonBay25}. In practical implementations, however,
	these correlations are vulnerable to environmental noise and decoherence, which progressively
	reduce their nonclassical character and constrain the performance of quantum protocols~\cite{Zur03,
		TinEbe09,DarSlu24}. Understanding how different forms of quantum correlations change under physical
	noise is therefore important both for fundamental studies of nonclassicality and for assessing the
	performance of realistic quantum technologies~\cite{StrHer11,IkeLow24}.
	
	Quantum correlations form a hierarchy of nonclassicality~\cite{WisDoh07,QuiBru15}, with Bell
	nonlocality~\cite{Bell64} implying Einstein-Podolsky-Rosen (EPR) steering~\cite{UolGuh20}, which in
	turn implies entanglement. This hierarchy has been demonstrated experimentally through increasingly
	stringent tests of nonclassical correlations, including loophole-free Bell
	experiments~\cite{HenHan15,GiuZei15,ShaNam15}. Related hierarchies have also been investigated for
	quantum correlations defined in temporal scenarios, where distinct forms of temporal quantum
	correlations can likewise be ordered~\cite{KuNor18}.
	
	Here, we consider the standard bipartite Bell--steering--entanglement hierarchy and focus on how
	these correlations degrade along physical trajectories in quantum state space. In particular, a
	noisy quantum state may lose Bell nonlocality while remaining steerable, subsequently lose steering
	while remaining entangled, and eventually become separable. A quantitative description that can be
	applied consistently across these different levels is therefore useful for studying the degradation
	of quantum correlations under noise.
	
	Robustness measures provide one such description by quantifying the amount of mixing with classical
	states required to remove a specified quantum resource. The robustness of entanglement was
	introduced in~\cite{VidTar99} and subsequently generalized in several directions~\cite{Steiner03},
	including formulations for steering~\cite{SkrCav14,PiaWat15}, Bell nonlocality and measurement
	incompatibility~\cite{CavSkr16}, as well as quantum
	channels~\cite{HarNie03,HorRus03,TakReg19,GouSca21}. These results establish robustness as a broadly
	applicable resource measure. Nevertheless, the behavior of robustness along continuous quantum
	dynamical trajectories, and its connection to the timescale over which a quantum resource is lost,
	have received comparatively less attention across the Bell--steering--entanglement hierarchy.
	
	In this work, we develop a robustness framework for bipartite quantum correlations based on the
	nested local-hidden-variable (LHV), local-hidden-state (LHS), and separable sets. We show that the
	corresponding robustness measures inherit an ordering from the inclusion relations among these sets.
	This provides a common quantitative framework for comparing the strength and degradation of Bell
	nonlocality, steering, and entanglement, complementing the criteria used to identify these
	resources.
	
	We study how these robustness measures evolve under quantum channels and, in particular, under
	continuous Markovian dynamics. For noise processes that can be expressed as mixing with a classical
	reference state, the robustness determines the noise threshold at which the evolving state enters
	the corresponding classical set. We use this relation to connect robustness with transition times
	and to define a time-dependent robustness profile. At the level of a Lindblad master equation, we
	further express the instantaneous change of robustness through the directional derivative of the
	robustness along the dynamical trajectory. This separates the effect of the physical evolution from
	the response of the chosen robustness measure. The resulting framework connects static robustness
	quantification with the continuous-in-time evolution of quantum resources under noise. It provides a
	way to characterize not only whether a resource survives a given channel, but also how its
	robustness changes during the evolution and when the corresponding classical regime is reached.
	
	Applications to isotropic depolarizing noise and to local amplitude damping dynamics recover the
	known hierarchy of transition thresholds and show how asymmetric dynamics can produce distinct
	steering robustness profiles and lifetimes for the two-party assignments.
	
	The paper is organized as follows. In \srf{sec:framework}, we introduce the robustness framework
	and establish the ordering of robustness measures associated with the hierarchy of quantum
	correlations. Section~\ref{sec:channel} develops the dynamical formulation of robustness and its
	connection to quantum channels and open-system evolution. Section~\ref{sec:iso:noise} applies the
	framework to isotropic depolarizing noise using two-qubit Werner states, while
	\srf{sec:directional} considers directional robustness under asymmetric dynamics. We conclude in
	\srf{sec:conclusion} with a summary of the main results.
	
	\section{Geometric Robustness Framework} \label{sec:framework}
	
	Quantum correlations form a hierarchy according to the assumptions imposed on the measurement
	devices~\cite{WisDoh07,QuiBru15}. These assumptions define distinct classical models, each of which
	excludes a corresponding class of quantum correlations that cannot be described within that model.
	In this work, we associate each classical model with a convex subset of the quantum state space and
	quantify the distance from these sets using entanglement robustness measures.
	
	\subsection{Hierarchy of Quantum Correlations} \label{sec:definitions}
	
	Consider a bipartite state $\rho_{AB}$ shared between two parties. The statistics generated by
	local measurements are described by
	\begin{equation}
		p(a,b|A,B) = \mathrm{Tr} \left[\left(\hat M_{a|A}\otimes \hat M_{b|B}\right)\rho_{AB} \right],
	\end{equation}
	where $A$ and $B$ label the measurement settings, $a$ and $b$ denote the
	corresponding outcomes, and $\hat M_{a|A}$ and $\hat M_{b|B}$ are positive local measurement
	operators satisfying $\sum_a \hat M_{a|A}=\hat I= \sum_b \hat M_{b|B}$, where $\hat I$ is the
	identity operator.
	
	Different assumptions concerning the characterization of the measurement
	devices lead to different classical descriptions of the observed
	correlations. In the fully device-independent setting, both measurement
	devices are treated as uncharacterized. The corresponding local hidden
	variable (LHV) model is defined by~\cite{UolGuh14}
	\begin{equation} 	\label{eqn:corr:lhv}
		p(a,b|A,B)=\sum_{\lambda}p(\lambda)\,p(a|A,\lambda)\,p(b|B,\lambda),
	\end{equation}
	where $\lambda$ is a shared classical random (hidden) variable distributed according to
	$p(\lambda)$. The notation $p(a|A,\lambda)$ and $p(b|B,\lambda)$ represents conditional probability
	distributions specifying the local response of each party to the measurement choice, given the
	shared variable $\lambda$. The set of states whose measurement statistics admit such a description
	is denoted by
	$\mathfrak C_{\mathrm{LHV}}$.
	
	If one subsystem is trusted while the other remains uncharacterized, the
	corresponding classical model is a local hidden state (LHS) model. For
	example, when Bob's subsystem is trusted, the correlations admit an LHS
	decomposition of the form
	\begin{equation} 	\label{eqn:corr:lhs}
		p(a,b|A,B) = \sum_{\lambda} p(\lambda)\,p(a|A,\lambda)\,\mathrm{Tr} \left[ \hat
		M_{b|B}\rho_B^\lambda \right],
	\end{equation}
	where $\rho_B^\lambda$ denotes a quantum state assigned to Bob's subsystem.
	The corresponding set is denoted by $\mathfrak C_{\mathrm{LHS}}$.
	Interchanging the trusted and untrusted parties gives, in general, a
	different LHS model and hence a directional notion of steering.
	
	Finally, when both subsystems are fully characterized, the relevant
	classical set is the set of separable states, $
	\rho_{AB}=\sum_{\lambda}p(\lambda)\,\rho_A^\lambda\otimes\rho_B^\lambda$,
	denoted by $\mathfrak C_{\mathrm{SEP}}$. The corresponding measurement
	statistics can equivalently be written as
	\begin{equation}
		p(a,b|A,B)=\sum_{\lambda} p(\lambda)\,\mathrm{Tr} \left[\hat M_{a|A}\rho_A^\lambda \right]
		\mathrm{Tr}\left[\hat M_{b|B}\rho_B^\lambda \right].
	\end{equation}
	These classical sets satisfy the inclusion hierarchy
	\begin{equation} 	\label{eqn:inclusion:set}
		\mathfrak C_{\mathrm{SEP}} \subsetneq \mathfrak C_{\mathrm{LHS}} \subsetneq \mathfrak
		C_{\mathrm{LHV}}.
	\end{equation}
	Thus, the different classes correspond to nested regions of the quantum
	state space, with entangled states lying outside $\mathfrak C_{\mathrm{SEP}}$,
	steerable states lying outside the appropriate LHS set, and Bell-nonlocal
	states lying outside $\mathfrak C_{\mathrm{LHV}}$.
	
	\subsection{Hierarchy of Robustness} \label{sec:hierarchy:R}
	
	The nested structure in Eq.~\eqref{eqn:inclusion:set} provides a natural
	setting for comparing the robustness associated with different classical
	sets. Let $\mathfrak C$ denote a convex set of states representing a specified classical model. For
	a fixed reference state $\rho_c\in\mathfrak C$, the relative robustness of a state $\rho$ is
	defined
	as~\cite{VidTar99}
	\begin{equation}
		R(\rho\|\rho_c) = \min \left\{\alpha\geq0 \;\middle|\;\rho(\alpha,\rho_c)\in\mathfrak C \right\},
	\end{equation}
	where
	\begin{equation} 	\label{eqn:rho:geo:mix}
		\rho(\alpha,\rho_c) \equiv \frac{\rho+\alpha\rho_c}{1+\alpha}
	\end{equation}
	is the corresponding convex mixture. The parameter $\alpha$ quantifies the minimum amount of the
	reference state required to bring $\rho$ into the classical set. Optimizing also over the choice of
	reference state gives the absolute robustness,
	\begin{equation} 	\label{eqn:abs:R}
		R_{\mathfrak C}(\rho) = \min_{\rho_c\in\mathfrak C} R(\rho\|\rho_c) =
		\min_{\substack{\alpha\geq0\\\rho_c\in\mathfrak C}}
		\left\{\alpha\;\middle|\;\rho(\alpha,\rho_c)\in\mathfrak C \right\}.
	\end{equation}
	The robustness therefore provides a geometric coordinate measuring the
	distance of a state from the convex set $\mathfrak C$ along the family of
	admissible mixing directions. If an optimal solution exists, the
	corresponding condition can also be expressed as the pseudomixture
	decomposition
	\begin{equation} 	\label{eqn:afine:decomp}
		\rho=(1+\alpha^*)\rho_{c_1}-\alpha^*\rho_{c_2},
	\end{equation}
	where $\alpha^*=R_{\mathfrak C}(\rho)$ and $\rho_{c_1},\rho_{c_2}\in\mathfrak C$.
	This representation makes explicit the interpretation of robustness as an affine extension of the
	classical set.
	
	The nesting of the classical sets immediately induces an ordering of the
	corresponding robustness measures.
	
	\begin{theorem}[Robustness ordering under nested classical sets] 	\label{theorem:hierarchy:R}
		Let $\mathfrak C_2\subsetneq\mathfrak C_1$ be two nested convex sets. Then
		\begin{equation}
			R_{\mathfrak C_2}(\rho) \geq R_{\mathfrak C_1}(\rho).
		\end{equation}
	\end{theorem}
	
	\begin{proof}
		Consider any feasible pair $(\alpha,\rho_c)$ in the optimization defining
		$R_{\mathfrak C_2}(\rho)$. Thus, $\rho_c\in\mathfrak C_2$ and
		$\rho(\alpha,\rho_c)\in\mathfrak C_2$. Since $\mathfrak C_2\subsetneq\mathfrak C_1$, the same pair
		is also feasible for the optimization defining
		$R_{\mathfrak C_1}(\rho)$. Hence the feasible set of the latter optimization
		contains that of the former, and its minimum cannot be larger. This completes the proof.
	\end{proof}
	Applying Theorem~\ref{theorem:hierarchy:R} directly to the set relation in \erf{eqn:inclusion:set}
	establishes
	\begin{equation} \label{eqn:hierarchy:R} 
		R_{\mathrm{LHV}}(\rho) \leq R_{\mathrm{LHS}}(\rho) \leq R_{\mathrm{SEP}}(\rho).
	\end{equation}
	Thus, the hierarchy of classical sets is reflected directly in the
	corresponding robustness measures.
	
	Equation~\eqref{eqn:hierarchy:R} is a statement about the hierarchy of robustness and does not by
	itself determine the dynamics of robustness.
	For a given quantum channel, the state follows a particular trajectory
	through the state space, and the rate at which the trajectory approaches a
	classical set depends on both the geometry of that set and the direction of
	the dynamical flow. For symmetric trajectories these effects can be
	separated particularly simply, as illustrated by the Werner-state example
	in \srf{sec:iso:noise}. For more general anisotropic dynamics, the trajectory can probe different
	robustness landscapes at different rates.
	
	Therefore, it is also useful to relate the robustness parameter to a physical noise
	fraction. Consider a mixing process of the form
	\begin{equation}
		\rho(\wp,\rho_c) = (1-\wp)\rho+\wp\rho_c.
	\end{equation}
	Comparing this expression with Eq.~\eqref{eqn:rho:geo:mix} gives $	\wp={\alpha}/({1+\alpha})$.
	Consequently, if the trajectory reaches the classical boundary at a robustness value
	$R(\rho\|\rho_c)$, the corresponding noise fraction is
	\begin{equation} 	\label{eqn:thresh:noise}
		\wp_{c}=\frac{R(\rho\|\rho_c)}{1+R(\rho\|\rho_c)}.
	\end{equation}
	This relation provides a direct connection between the  robustness
	and a physically realized mixing parameter. Once the channel determines
	$\wp$ as a function of time, distance, or another evolution parameter,
	Eq.~\eqref{eqn:thresh:noise} identifies the point at which the trajectory
	reaches the corresponding classical set.

	\section{Robustness under Quantum Dynamics} \label{sec:channel}
	
	The robustness measures introduced in the previous section quantify the
	nonclassical character of an individual quantum state relative to a
	specified classical set. In a physical setting, however, quantum
	states are continuously modified by interactions with their environment.
	It is therefore natural to consider robustness as a dynamical quantity and
	to study how it changes under the action of a physical quantum channel.
	
	Let $\mathcal E$ denote a completely positive trace-preserving (CPTP) map
	acting on bipartite states, $\rho \longmapsto \mathcal E(\rho)$.
	For a given convex classical set $\mathfrak C$, the map $\mathcal E$ is $\mathfrak C$-preserving if
	\begin{equation}
		\mathcal E(\mathfrak C)\subseteq\mathfrak C.
	\end{equation}
	Such a channel cannot generate the corresponding nonclassical resource
	from a state that already belongs to $\mathfrak C$. Many physically relevant operations and noise
	processes satisfy this property for the appropriate choice of $\mathfrak C$. The robustness is
	non-increasing under every $\mathfrak C$-preserving channel.
	
	\begin{theorem}[Monotonicity under $\mathfrak C$-preserving channels]  
		\label{theorem:channel:monotonicity}
		Let $\mathfrak C$ be a convex classical set and let $\mathcal E$ be a
		CPTP map satisfying $\mathcal E(\mathfrak C)\subseteq\mathfrak C$. Then
		\begin{equation}
			R_{\mathfrak C}\!\left(\mathcal E(\rho)\right) \le R_{\mathfrak C}(\rho)
		\end{equation}
		for every bipartite quantum state $\rho$.
	\end{theorem}
	
	\begin{proof}
		Let \erf{eqn:afine:decomp}
		denote an optimal robustness decomposition. Applying the linear map
		$\mathcal E$ gives
		\begin{equation}
			\mathcal E(\rho)=(1+\alpha^*)\mathcal E(\rho_{c_1})-\alpha^*\mathcal E(\rho_{c_2}).
		\end{equation}
		Since $\mathcal E$ is $\mathfrak C$-preserving, $\mathcal E(\rho_{c_1}),\mathcal
		E(\rho_{c_2})\in\mathfrak C$. Hence $\alpha^*$ remains a feasible value in the robustness
		optimization for the output state. It follows that
		\begin{equation}
			R_{\mathfrak C}\!\left(\mathcal E(\rho)\right) \leq \alpha^* = R_{\mathfrak C}(\rho).
		\end{equation}
	\end{proof}
	
	Theorem~\ref{theorem:channel:monotonicity} establishes robustness as a resource monotone under
	$\mathfrak C$-preserving quantum operations. In particular, the same
	result applies to the classical sets associated with Bell nonlocality,
	EPR steering, and entanglement. It therefore provides a common framework
	for studying the evolution of these resources under physical noise,
	without requiring the dynamics to be characterized through a particular
	Bell inequality, steering criterion, or entanglement witness.
	
	\subsection{Robustness profiles and transition times}  \label{sec:robustness:profile}
	
	Consider a continuous family of quantum channels
	$\{\mathcal E_t\}_{t\geq0}$ generating the state trajectory $\rho(t)=\mathcal E_t(\rho_0)$. The
	robustness evaluated along this trajectory defines the
	\emph{robustness profile}
	\begin{equation} 	\label{eqn:robustness:profile}
		R_{\mathfrak C}(t) \coloneq R_{\mathfrak C}\!\left(\rho(t)\right).
	\end{equation}
	The profile describes the evolution of the nonclassical resource as the
	state is transformed by the physical dynamics. Its functional form
	depends both on the robustness functional associated with $\mathfrak C$
	and on the particular state-space trajectory generated by the channel.
	
	Whenever $R_{\mathfrak C}$ is differentiable along the trajectory, we
	define the instantaneous robustness decay rate by
	\begin{equation} 	\label{eqn:robustness:decay}
		\Gamma_{\mathfrak C}(t) \coloneq -\dot{R}_{\mathfrak C}(t).
	\end{equation}
	For a $\mathfrak C$-preserving evolution, the monotonicity theorem implies
	that the robustness is non-increasing, and therefore
	$\Gamma_{\mathfrak C}(t)\geq0$\footnote{For non-Markovian dynamics, temporary revivals associated
		with information backflow maybe exhibited, leading to intervals where the decay rate vanishes or
		even becomes negative.}.
	
	A corresponding transition time is defined as the first time at which the state reaches the
	classical set,
	\begin{equation} 	\label{eqn:transition:time}
		t_{\mathfrak C} =\inf \left\{t\geq0 \,\middle|\, R_{\mathfrak C}(t)=0 \right\}.
	\end{equation}
	Thus, the robustness profile contains information about both the instantaneous evolution of the
	resource and the point at which the corresponding classical boundary is reached.
	
	\subsection{Markovian open-system dynamics}
	\label{sec:OQS}
	
	For Markovian open-system dynamics, the state evolution is governed by the
	Gorini--Kossakowski--Sudarshan--Lindblad master equation~\cite{Lindblad76,GKS76},
	\begin{equation} 	\label{eqn:GKSL}
		\dot{\rho}(t)=\mathcal L(\rho(t)),
	\end{equation}
	where the Lindblad generator has the standard form
	\begin{equation} 	\label{eq:lindblad:generator}
		\mathcal L(\rho) = -i[\hat H,\rho] +\sum_k \left(\hat L_k\rho \hat L_k^\dagger-\frac{1}{2}\{\hat
		L_k^\dagger \hat L_k,\rho\}\right).
	\end{equation}
	Here $\hat H$ is the system Hamiltonian and $\hat L_k$ are the Lindblad jump
	operators describing the dissipative processes. The corresponding
	quantum dynamical semigroup is formally written as $\rho(t)=e^{t\mathcal L}\rho(0)$. The Lindblad
	generator therefore specifies the instantaneous direction
	of motion of the density operator in state space. This provides the
	natural dynamical description for evaluating how a robustness functional
	changes under a specified physical noise process.
	
	\subsubsection{Lindblad dynamics for state mixing}
	\label{app:affine:noise}
	
	The state-mixing noise model introduced in Sec.~\ref{sec:hierarchy:R}
	can be realized by a Markovian process. This provides an explicit Lindblad realization of the
	state-space interpolation between a resource state and
	a reference state. For a product reference state $	\rho_c=\rho_c^A\otimes\rho_c^B$,
	with spectral decompositions
	\begin{equation}
		\rho_c^A=\sum_k p_k\ket{a_k}\bra{a_k}, \qquad \rho_c^B=\sum_m q_m\ket{b_m}\bra{b_m},
	\end{equation}
	where $\{\ket{a_k}\}$ and $\{\ket{b_m}\}$ form orthonormal bases for the local Hilbert spaces of
	Alice and Bob, respectively, and $\sum_k p_k = \sum_m q_m = 1$.
	one may choose local jump operators of the form
	\begin{align} 
		\hat L_{jk}^{A} &= \sqrt{\gamma\, p_k}\,\ket{a_k}\bra{a_j}\otimes \hat I_B,
		\label{affine:noise:La}\\
		\hat L_{jm}^{B} &= \sqrt{\gamma\, q_m}\,\hat I_A\otimes\ket{b_m}\bra{b_j}, \label{affine:noise:Lb}
	\end{align}
	with $\gamma$ denoting a characteristic rate. These operators drive the local subsystems towards
	the
	target states $\rho_c^A$ and $\rho_c^B$, respectively.
	
	When the initial correlated state has the same local marginals as the target product state, the
	resulting master equation reduces to
	\begin{equation} \label{eqn:ME:affine:noise}
		\dot{\rho}(t) = -2\gamma\left[\rho(t)-\rho_c\right].
	\end{equation}
	Its solution is
	\begin{equation} 	\label{eqn:affine:evolution}
		\rho(t) = [1-\lambda(t)]\rho_0+\lambda(t)\rho_c,\qquad \lambda(t)=1-e^{-2\gamma t}.
	\end{equation}
	Thus, the state-mixing model used in the robustness construction can be
	generated by a standard Markovian Lindblad process. The detailed evaluation of the corresponding
	jump terms is given in Appendix~\ref{app:affine:derivation}.

	\subsubsection{Interpretation of robustness decay}
	\label{sec:interpretation:decay}
	
	The robustness profile can be related directly to the generator of the
	underlying open-system dynamics. Suppose that the robustness functional
	is differentiable\footnote{At points where the robustness functional is not differentiable, the
		ordinary Fr\'echet derivative in Eq.~\eqref{eq:robustness:chainrule} is not defined and the
		robustness profile should instead be treated using one-sided or generalized derivatives. The
		results
		below therefore assume differentiability wherever an explicit robustness decay rate is reported.}
	at
	$\rho(t)$. Using \erf{eqn:GKSL} and applying the chain rule gives
	\begin{equation} 	\label{eq:robustness:chainrule}
		\dot{R}_{\mathfrak C}(\rho(t)) = DR_{\mathfrak C}(\rho(t)) \left[\mathcal L(\rho(t)) \right],
	\end{equation}
	where $DR_{\mathfrak C}(\rho)$ denotes the functional derivative of the
	robustness evaluated at $\rho$ (see Appendix~(\ref{appn:frechet:der}) for more details).
	Equation~\eqref{eq:robustness:chainrule} expresses the instantaneous
	change of robustness as a directional derivative along the physical
	state-space trajectory generated by the Lindblad equation. The Lindblad generator determines the
	direction of the state evolution, while the functional derivative determines the response of the
	robustness to that particular perturbation.
	
	This relation does not require an explicit analytical expression for the
	full robustness functional. In general, such an expression may be
	difficult to obtain because robustness is defined through a convex
	optimization problem. The directional derivative is therefore most
	useful when the state trajectory possesses additional structure that
	allows the robustness to be evaluated or reduced to a lower-dimensional
	family. 
	
	For finite-dimensional quantum systems, the derivative can be
	represented using the Hilbert--Schmidt inner product $\langle A,B\rangle_{\rm
		HS}=\mathrm{Tr}(A^\dagger B)$. The Riesz representation theorem guarantees that every bounded
	linear
	functional admits a unique representation as an inner product with an operator~\cite{Ste08}.
	Accordingly, there exists a unique operator $\nabla R_{\mathfrak C}(\rho)$ such that
	\begin{equation} 	\label{eq:frechet:gradient:HS}
		DR_{\mathfrak C}(\rho)[X] = \left\langle \nabla R_{\mathfrak C}(\rho),X\right\rangle_{\rm HS}.
	\end{equation}
	The robustness rate can therefore be written as
	\begin{equation} 	\label{eq:robustness:gradient:flow}
		\dot{R}_{\mathfrak C} =\left\langle \nabla R_{\mathfrak C}(\rho),\mathcal L(\rho)
		\right\rangle_{\rm HS}.
	\end{equation}
	This expression provides an operator-space interpretation of the
	directional derivative. The robustness gradient describes the local
	sensitivity of the robustness functional to changes in the state, whereas
	$\mathcal L(\rho)$ specifies the instantaneous state change generated by
	the physical dynamics. The robustness decay is consequently determined
	by their Hilbert--Schmidt overlap rather than by the magnitude of the
	Lindblad evolution alone.
	
	This interpretation allows one to derive general bounds on the instantaneous variation of the
	robustness. Applying the Cauchy--Schwarz inequality gives the general bound
	\begin{equation} 	\label{eq:robustness:CS:bound}
		\left|\dot{R}_{\mathfrak C}\right|\leq\left\|\nabla R_{\mathfrak C}(\rho)\right\|_{\rm
			HS}\left\|\mathcal L(\rho)\right\|_{\rm HS}.
	\end{equation}
	The first factor characterizes the local sensitivity of the robustness
	functional, while the second quantifies the instantaneous magnitude of
	the physical state-space evolution. Equation~\eqref{eq:robustness:CS:bound}
	therefore separates the response of the robustness functional from the
	strength of the underlying dynamical evolution.

	\section{Isotropic noise and robustness decay}  \label{sec:iso:noise}

Having established the robustness hierarchy and its dynamical formulation, we now consider a
canonical example in which the relevant transition boundaries are known analytically. Two-qubit
Werner states~\cite{Werner89} provide a natural benchmark because they form a one-parameter family
that is preserved under isotropic depolarizing dynamics and has well-characterized thresholds for
entanglement, steering, and Bell nonlocality. The purpose of this example is therefore twofold:
first, to recover the known transition points within the robustness framework; and second, to
demonstrate how the directional derivative of robustness can be evaluated explicitly when the
Lindblad trajectory is confined to a one-dimensional state manifold. We consider the Werner family
	\begin{equation} 	\label{eqn:2q:werner}
		\rho_{\rm W}(\wp)=(1-\wp)\rho_- + \wp\,\varrho_{\rm c},
	\end{equation}
	where $\rho_- = |\psi_-\rangle\langle\psi_-|$ is the singlet state, $\varrho_{\rm c}=\hat I/4$ is
the maximally mixed state, and $\wp\in[0,1]$ denotes the isotropic noise fraction. The full range
for which $\rho_{\rm W}(\wp)$ is a valid density operator extends beyond this interval, but the
above restriction for$\wp$ is the physically relevant range for the present interpretation as
mixing with white noise.

For Werner states, the critical noise thresholds associated with Bell nonlocality, steering, and
entanglement are known exactly (or, in the steering case, depend on the measurement
scenario)~\cite{Werner89,Horodecki95, Horodecki96,BruWeh14}. In particular,
$\wp_{\mathrm{LHV}}=1-{1}/{\sqrt2}$ and $\wp_{\mathrm{SEP}}=2/3$ while the steering threshold
satisfies
	\begin{equation} 	\label{eqn:2q:werner:p:order}
		\wp_{\mathrm{LHV}} < \wp_{\mathrm{LHS}} < \wp_{\mathrm{SEP}}.
	\end{equation}
	Within the robustness formulation, these thresholds correspond to the
	values of the Werner parameter at which the state enters the corresponding
	classical sets $\mathfrak C_{\mathrm{LHV}}$, $\mathfrak C_{\mathrm{LHS}}$, and
	$\mathfrak C_{\mathrm{SEP}}$. The symmetry reduction discussed in
	Appendix~\ref{app:werner:symmetry}
	further shows that the absolute robustness optimization for this family is realized by mixing with
	the maximally mixed state. Thus, the robustness along the Werner family can be related directly to
	the corresponding threshold parameter.
	
	\begin{figure*}[t]
		\centering
		\includegraphics[width=\linewidth]{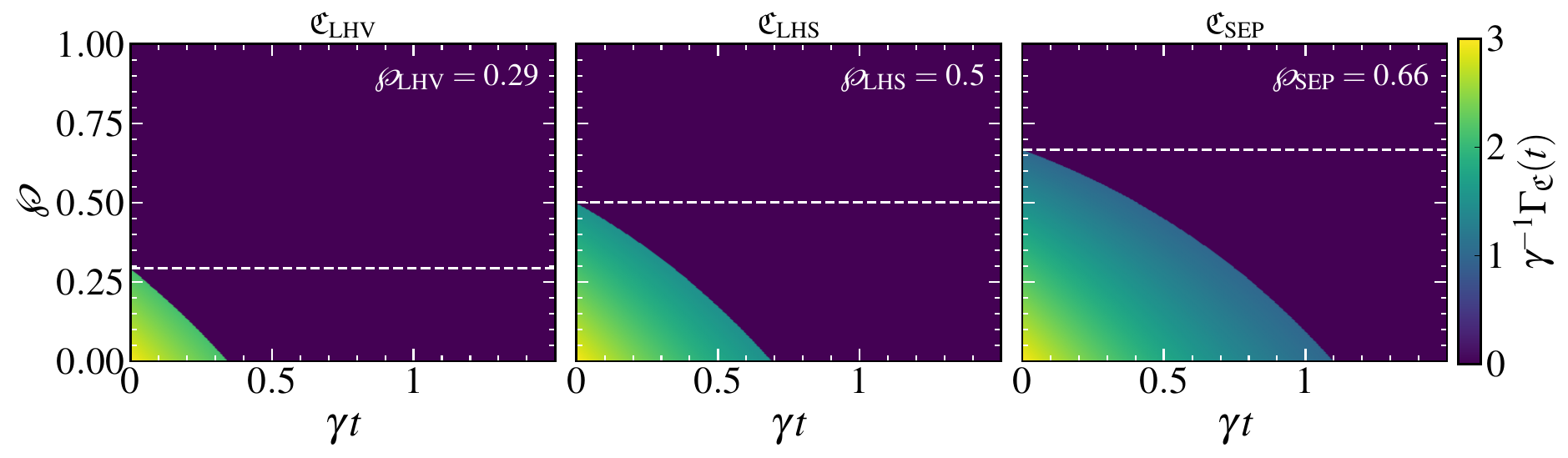}
		\caption{\label{fig:decay}
			(Color online). Dimensionless robustness decay rate $\gamma^{-1}\Gamma_{\mathfrak C}(t)$ for
Werner states evolving under exponential depolarizing noise, shown as a function of the initial
parameter $\wp$ and dimensionless time $\gamma t$. The three panels correspond to the classical
sets $\mathfrak C_{\mathrm{LHV}}$ (left), $\mathfrak C_{\mathrm{LHS}}$ (center), and $\mathfrak
C_{\mathrm{SEP}}$ (right). The colored regions indicate points along the evolution for which the
corresponding robustness remains positive, while the dark regions indicate that the robustness
has already reached zero. The curved boundary separating the two regions is the transition-time
curve defined by Eq.~\eqref{eqn:w:trantime}. The horizontal dashed line marks the critical Werner
parameter $\wp_{\mathfrak C}$, above which the corresponding robustness is already zero at the
initial time. The color scale gives the instantaneous robustness decay rate from
Eq.~\eqref{eqn:werner:decay:exp}. Although the depolarizing dynamics are identical in all three
panels, the decay rate is rescaled by the factor $(1-\wp_{\mathfrak C})^{-1}$, reflecting the
different robustness thresholds associated with the three classical sets.
		}
	\end{figure*}

	\subsection{Depolarizing Lindblad dynamics}
	
	We now consider the Markovian depolarizing evolution similar to \erf{eqn:affine:evolution} where $
	\mathcal E_t(\rho) = [1-\Lambda(t)]\rho + \Lambda(t)\varrho_{\rm c}$.  This map can be generated by
	considering the master equation (\ref{eqn:GKSL}), and choosing the following jump operators
	\begin{align}
		\hat L_j^{A} &= \sqrt{\gamma_{\rm d}}\,\hat\sigma_j\otimes\hat I_B,
		\\
		\hat L_j^{B} &= \sqrt{\gamma_{\rm d}}\,\hat I_A\otimes\hat\sigma_j,
	\end{align}
	with $j=1,2,3$. Here $\gamma_{\rm d}$ corresponds to the depolarization rate and $\hat{\sigma}_j$
	denotes the Pauli operators. With the maximally mixed state as the stationary state, the resulting
	Lindblad dynamics generates the isotropic contraction of the two-qubit state towards $\varrho_{\rm
		c}$. Writing
	$\gamma \equiv 8\gamma_{\rm d}$ for the effective depolarization rate, the resulting
	noise parameter is given by $\Lambda(t)=1-e^{-\gamma t}$.
	
	Because the depolarizing dynamics preserve the Werner family, an initial
	state $\rho_0=\rho_{\rm W}(\wp)$ remains on the same one-dimensional
	manifold, $	\rho(t)=\rho_{\rm W}(\wp_{\rm eff}(t))$, with (consult Appendix~\ref{app:werner:depol}
	for the details)
	\begin{equation} 	\label{eqn:werner:effective:parameter}
		\wp_{\rm eff}(t) = \Lambda(t)+[1-\Lambda(t)]\wp = 1-(1-\wp)e^{-\gamma t}.
	\end{equation}
	The state trajectory is therefore completely characterized by the single
	dynamical coordinate $\wp_{\rm eff}(t)$.
	
	The transition associated with a classical set $\mathfrak C$ occurs when
	the trajectory reaches its corresponding Werner threshold, $\wp_{\rm eff}(t_{\mathfrak
		C})=\wp_{\mathfrak C}$. For the exponential dynamics, this gives
	\begin{equation} 	\label{eqn:w:trantime}
		t_{\mathfrak C} = \frac{1}{\gamma}\ln\left(\frac{1-\wp}{1-\wp_{\mathfrak C}}\right),
	\end{equation}
	provided $\wp\leq\wp_{\mathfrak C}$. If $\wp>\wp_{\mathfrak C}$, the initial state is already
	inside
	the corresponding classical set and $t_{\mathfrak C}=0$.
	
	Equation~\eqref{eqn:w:trantime} illustrates the separation between the
	state-space criterion and its dynamical parameterization. The value
	$\wp_{\mathfrak C}$ specifies where the classical boundary intersects the
	Werner manifold, whereas the Lindblad dynamics determine how rapidly the
	trajectory reaches that point. For a fixed depolarizing rate, states with
	smaller initial $\wp$ therefore remain outside the corresponding classical
	set for longer. In particular,
	\begin{equation} 	\label{eqn:w:t:trans:ub}
		0 \leq t_{\mathfrak C} \leq \frac{1}{\gamma} \ln \left(\frac{1}{1-\wp_{\mathfrak C}} \right),
	\end{equation}
	where the upper bound is attained by the singlet state. Therefore, the ordering of threshold
	parameters (\ref{eqn:2q:werner:p:order}) consequently gives
	\begin{equation}
		t_{\mathrm{LHV}} \leq t_{\mathrm{LHS}} \leq t_{\mathrm{SEP}},
	\end{equation}
	for the same initial Werner state and depolarizing dynamics.
	
	\subsubsection{Robustness decay along the Lindblad trajectory}
	
	The transition time characterizes when the robustness reaches zero, but the same trajectory also
allows its instantaneous variation to be evaluated, \erf{eq:robustness:chainrule}. For the present
example, the full functional derivative of the robustness over the entire operator space is not
required. Since the Lindblad dynamics preserve the Werner family, $\mathcal L(\rho(t))$ is tangent
to the one-dimensional Werner trajectory. It is therefore sufficient to evaluate the directional
derivative of the robustness along this trajectory. A straightforward calculation shows that
(Appendix~\ref{app:werner:decay}):
	\begin{equation} 	\label{eqn:werner:robustness:profile}
		R_{\mathfrak C}\left(\rho_{\mathrm W}(\wp_{\mathrm{eff}})\right)=\frac{\wp_{\mathfrak
				C}-\wp_{\mathrm{eff}}}{1-\wp_{\mathfrak C}}.
	\end{equation}
	Thus, along the Lindblad trajectory, the robustness decay can be expressed as
	\begin{equation} 	\label{eqn:werner:decay:general}
		\Gamma_{\mathfrak C}(t) = -DR_{\mathfrak C}(\rho(t)) \left[\mathcal L(\rho(t))\right]=
		-\frac{dR_{\mathfrak C}} {d\wp_{\rm eff}} \dot{\wp}_{\rm eff},
	\end{equation}
	and using \erfa{eqn:werner:effective:parameter}{eqn:werner:robustness:profile} it simplifies to
	\begin{equation} 	\label{eqn:werner:decay:exp}
		\Gamma_{\mathfrak C}(t)=\frac{\gamma(1-\wp)e^{-\gamma t}}{1-\wp_{\mathfrak C}}.
	\end{equation}
	This expression is valid while
	$\wp_{\rm eff}(t)<\wp_{\mathfrak C}$. Once the trajectory enters the
	classical set, the robustness is zero and its decay rate is taken to vanish.
	
	Equation~\eqref{eqn:werner:decay:exp} provides a direct illustration of the
	directional-derivative formulation. The factor $\dot{\wp}_{\rm eff}$ is determined entirely by the
	Lindblad dynamics, whereas the factor $(1-\wp_{\mathfrak C})^{-1}$ arises from the restriction of
	the robustness functional to the Werner manifold. Thus, even for identical physical dynamics,
	different classical boundaries produce different instantaneous robustness decay profiles.
	
	Figure~\ref{fig:decay} summarizes these results. The boundary between the
	colored and dark regions corresponds to the transition time
	$t_{\mathfrak C}$, while the color scale represents the instantaneous
	robustness decay rate. The horizontal dashed line marks
	$\wp=\wp_{\mathfrak C}$, above which the initial state already has zero
	robustness with respect to the corresponding classical set. The three
	panels therefore display the same physical trajectory under different robustness
	criteria, illustrating how the choice of classical set affects both the
	transition point and the instantaneous rate of robustness depletion.
	
	\section{Directional Robustness under Asymmetric Dynamics}   \label{sec:directional}
	
	The Werner-state example considered above involves isotropic depolarizing dynamics, for which the
state remains within the Werner family and the robustness can be evaluated using a single effective
noise parameter. More general open-system dynamics need not preserve such a one-parameter family.
In particular, nonunital channels can drive the state along trajectories that depend on the
subsystem on which the noise acts.

This provides a natural setting in which to examine the directional character of robustness. While
Bell nonlocality and entanglement are symmetric properties of a bipartite state, steering is
associated with a choice of trusted subsystem. Consequently, the same physical trajectory can be
evaluated with respect to different LHS robustness functionals. Hence, the resulting robustness
profiles and their decay rates need not coincide.
	
	\subsection{Direction-dependent robustness profile} \label{sec:geom:traj}
	
	A general single-qubit channel can be represented in the Bloch-vector
	picture as
	\begin{equation}
		\mathbf r \longmapsto C\mathbf r+\mathbf d,
	\end{equation}
	where $C$ describes the linear contraction of the Bloch vector and
	$\mathbf d$ accounts for nonunital displacement~\cite{NieChu10}. Unital processes satisfy $\mathbf
	d=0$, whereas nonunital processes drive the local state towards a
	preferred state.
	
	When such a process acts locally on one subsystem of a bipartite state, the
	resulting evolution generally differs from an isotropic (depolarizing)
	evolution. More importantly for steering, the robustness itself depends
	on the assignment of trusted and untrusted subsystems. We therefore denote
	the two directional steering robustness measures by $R_{\mathrm{LHS}}^{A\to B}(\rho)$, and
	$R_{\mathrm{LHS}}^{B\to A}(\rho)$, where the arrow specifies the direction of the steering task,
	and
	in general, $R_{\mathrm{LHS}}^{A\to B}(\rho)\neq R_{\mathrm{LHS}}^{B\to A}(\rho)$.
	
	The asymmetry of the robustness therefore does not require an asymmetric
	channel. Rather, an asymmetric channel provides a physical trajectory
	along which the two directional robustness functionals are sampled
	differently. For a common state evolution $\rho(t)$, one may consequently
	have different robustness for $A \rightleftarrows B$, with corresponding differences in their decay
	rates and transition times.
	
	The directional robustness rate is naturally described by the Lindblad formulation
	introduced in Sec.~\ref{sec:OQS}. For a differentiable robustness
	functional,
	\begin{equation}
		\dot R_{\mathrm{LHS}}^{A\to B}(t) = DR_{\mathrm{LHS}}^{A\to B}(\rho(t)) \left[\mathcal
		L(\rho(t))\right],
	\end{equation}
	and similarly for the opposite direction. Thus, the physical channel
	determines the state-space dynamic through $\mathcal L$, while the
	directional robustness functional determines the response of the resource
	measure to that evolution.
	
	\subsection{Amplitude-Damping Example}		\label{sec:amp:damp}
	
	To illustrate this distinction, we consider a local amplitude-damping
	process acting only on Alice's subsystem. In contrast to the depolarizing
	process considered in Sec.~\ref{sec:iso:noise}, amplitude damping is
	nonunital and drives Alice towards a preferred local state.
	
	For a Markovian amplitude-damping process, the Lindblad generator is
	obtained by choosing the jump operator according to
	\begin{equation} 	\label{eqn:jump:amp:damp}
		\hat L_{\mathrm a} = \sqrt{\gamma_{\rm a}}\, |0\rangle\langle1|\otimes\hat I_B,
	\end{equation}
	with vanishing Hamiltonian contribution. Solving the master equation gives the standard local
	amplitude-damping channel, $\rho(t) = [\mathcal E_A^{(t)}\otimes\mathcal I_B](\rho_0)$, with Kraus
	operators
	\begin{align}
		\hat K_0(t) &=|0\rangle\langle0|+\sqrt{1-\eta(t)}\,|1\rangle\langle1|, \label{eqn:k0:amp:damp}\\
		\hat K_1(t) &= \sqrt{\eta(t)}\,|0\rangle\langle1|,\label{eqn:k1:amp:damp}
	\end{align}
	where
	\begin{equation} 	\label{eqn:eta:amp:damp}
		\eta(t)=1-e^{-\gamma_{\rm a}t}.
	\end{equation}
	Take the initial state to be the maximally entangled singlet $\rho_0= \rho_-$, the resulting state
	is
	\begin{equation} 	\label{eqn:amp:damp:state}
		\rho(\eta) = \frac{1}{2}\left[\eta \op{00}{00}-\sqrt{1-\eta}\,\left(\op{01}{10}+ {\rm
			H.c.}\right)\right],
	\end{equation}
	where H.c. stands for Hermitian conjugate, and the reduced states are
	\begin{align} 	
		\rho_A(t) &= \frac12 \left([1+\eta(t)]|0\rangle
		\langle0|+[1-\eta(t)]|1\rangle\langle1|\right),\label{eqn:amp:damp:marginal:A}\\
		\rho_B(t) &= \frac12\hat I. \label{eqn:amp:damp:marginal:B}
	\end{align}
	Hence Alice's local state is progressively driven towards the ground
	state, whereas Bob's reduced state remains maximally mixed. The resulting
	bipartite trajectory is therefore qualitatively different from the
	isotropic Werner trajectory considered previously.
	
	\subsubsection{Directional robustness decay along the amplitude-damping trajectory}
	\label{sec:amp:damp:decay}
	
	The robustness dynamics can now be obtained using the directional
	derivative introduced in Sec.~\ref{sec:OQS}. Unlike the
	Werner family, the amplitude-damped states do not in general admit a
	robustness expression determined solely by an isotropic mixing parameter.
	Nevertheless, the physical trajectory is explicitly parameterized by
	$\eta(t)$, so that the directional derivative along the trajectory can be
	written as
	\begin{equation} 	\label{eqn:amp:damp:directional}
		\dot R_{\mathfrak C}(t) = DR_{\mathfrak C}(\rho(t)) \left[\mathcal L_{\mathrm a}(\rho(t))\right] =
		\frac{d R_{\mathfrak C}(\rho(\eta))}{d\eta} \,\dot\eta(t),
	\end{equation}
	and the corresponding decay rate is obtained through $\Gamma_{\mathfrak C}(t) =
	-\dot R_{\mathfrak
		C}(t)$. For the Markovian amplitude-damping process, $\dot\eta(t) = \gamma_{\rm
		a}[1-\eta(t)]$.
	Thus, the physical rate of motion along the trajectory is fixed by the Lindblad
	generator, whereas
	the factor
	$dR_{\mathfrak C}/d\eta$ describes the response of the robustness functional to
	that motion.
	
	This expression also makes clear an important distinction from the Werner example. For the Werner
family, symmetry reduces the robustness optimization to a known one-dimensional profile, allowing
the directional derivative to be evaluated analytically. For the amplitude-damping trajectory, the
corresponding derivative generally remains tied to the underlying robustness optimization. The
Lindblad formulation nevertheless identifies precisely the quantity that must be evaluated: the
action of the robustness functional on the physical tangent direction $\mathcal L_{\mathrm
	a}(\rho)$, rather than the full functional gradient.
	
	For steering, the same trajectory can be evaluated for the two possible
	trusted-party assignments. We therefore have
	\begin{equation} 	\label{eqn:amp:damp:directional:rates}
		\Gamma_{\mathrm{LHS}}^{A\rightleftarrows B}(t) = -DR_{\mathrm{LHS}}^{A\rightleftarrows B}(\rho(t))
		\left[\mathcal L_{\mathrm a}(\rho(t))\right].
	\end{equation}
	The channel is the same in both expressions; the difference arises from
	the directionality in robustness. This provides a direct
	generator-level description of how asymmetric dynamics can produce
	different steering robustness profiles.
	
	\subsubsection{Directional transition times}
	\label{sec:amp:damp:transition}
	
	The transition time is determined by the first value of $\eta$ for which
	the evolving state enters the relevant classical set. Denoting the
	corresponding value by $\eta_{\mathfrak C}$, which is defined according to $R_{\mathfrak
		C}\bigl(\rho(\eta_{\mathfrak C}(t))\bigr)=0$ the Markovian
	relation in Eq.~\eqref{eqn:eta:amp:damp} gives
	\begin{equation} 	\label{eqn:amp:damp:transition}
		t_{\mathfrak C}=\frac{1}{\gamma_{\rm a}}\ln\left(\frac{1}{1-\eta_{\mathfrak C}} \right).
	\end{equation}
	For the Bell and separability boundaries considered here, the critical values are
	$\eta_{\mathrm{LHV}}=1/2$ and $\eta_{\mathrm{SEP}}=1$, respectively (see
	Appendix~\ref{app:asym:thresholds}). The steering
	thresholds depend on the trusted-party assignment and measurement scenario.
	In general, the two directional thresholds
	$\eta_{\mathrm{LHS}}^{A\to B}$ and $\eta_{\mathrm{LHS}}^{B\to A}$ need not coincide,
	while both lie between the Bell and separability thresholds for the scenarios
	considered. Thus,
	\begin{equation}
		\eta_{\mathrm{LHV}} \leq \eta_{\mathrm{LHS}}^{A\to B},\eta_{\mathrm{LHS}}^{B\to A} \leq
		\eta_{\mathrm{SEP}},
	\end{equation}
	with their relative ordering depending on the steering criterion and
	measurement scenario. Since the transition time is a monotonically increasing
	function of the critical damping parameter, the corresponding time thresholds
	satisfy
	\begin{equation}
		t_{\mathrm{LHV}} \leq t_{\mathrm{LHS}}^{A\to B},t_{\mathrm{LHS}}^{B\to A} \leq t_{\mathrm{SEP}}.
	\end{equation}
	
	The amplitude-damping example therefore complements the isotropic
	depolarizing example in two ways. First, it demonstrates that the
	Lindblad--robustness relation remains applicable when the state trajectory
	does not coincide with a simple isotropic mixing path. Second, because the
	LHS robustness depends on the trusted-party assignment, the same physical
	trajectory can be associated with distinct robustness decay rates and
	transition times.

\section{Conclusion}		\label{sec:conclusion}
	
Understanding how nonclassical quantum resources degrade under realistic open-system dynamics is
important for characterizing the performance and lifetime of quantum protocols. In this work, we
developed a dynamical framework based on convex robustness measures, treating robustness not only
as a state-dependent quantifier of nonclassicality but also as a quantity that evolves along
physical quantum trajectories. The formalism connects the robustness of a state directly to the
Lindblad generator governing its evolution, providing a description of instantaneous robustness
decay and resource transition times.

We demonstrated the framework for two complementary noise processes. Isotropic depolarizing
dynamics provide an analytically tractable benchmark through the Werner family, while local
amplitude damping illustrates its extension to asymmetric, nonunital trajectories. In the latter
case, the directional nature of steering leads to distinct robustness profiles and transition times
for different trusted subsystems.

An important direction for future work is the efficient evaluation of robustness derivatives for
general optimization-defined resource measures, including at points where the robustness is not
differentiable. A related direction is to investigate how different physical processes contribute
to the change in robustness over time, particularly through their action on the robustness
gradient. Extending the framework to time-dependent and non-Markovian dynamics, correlated noise,
and higher-dimensional systems would further clarify how quantum resources evolve along realistic
open-system trajectories.

	% -------------------------------------------------------------------------
	%				acknowledgement
	% -------------------------------------------------------------------------
	\section*{Acknowledgment}
	This work was supported in part by the Australian Government
	through the Trailblazer Program.
	
	% ---------------------------------------------------
	%			Appendices
	% ---------------------------------------------------
	\appendix
	\numberwithin{equation}{section}
	
	\section{Derivation of ME (\ref{eqn:ME:affine:noise})}  \label{app:affine:derivation}
	
	The total Markovian master equation governed by the jump operators,
	\erfa{affine:noise:La}{affine:noise:Lb}, reads
	\begin{equation} 	\label{appn:ME:affine:noise}
		\dot{\rho}(t) = \sum_{j,k} \mathcal{D}\Big[\hat{L}_{j k}^A\Big]\rho(t) + \sum_{j,m}
		\mathcal{D}\Big[\hat{L}_{j m}^B\Big]\rho(t),
	\end{equation}
	where $\mathcal{D}[\hat{L}]\rho = \hat{L}\rho\hat{L}^\dagger - \frac{1}{2}\{\hat{L}^\dagger\hat{L},
	\rho\}$. Evaluating the jump terms for Alice's subsystem yields
	\begin{equation}
		\sum_{j,k} \hat{L}_{j k}^A \rho(t) \hat{L}_{j k}^{A\dagger} = \gamma \, \rho_c^A \otimes
		\mathrm{Tr}_A[\rho(t)].
	\end{equation}
	To calculate the corresponding anti-commutator loss term, using the completeness relation $\sum_j
	\ket{a_j}\bra{a_j} = \hat{I}_A$ we obtain:
	\begin{equation}
		\sum_{j,k} \hat{L}_{j k}^{A\dagger} \hat{L}_{j k}^A = \gamma \left( \hat{I}_A \otimes \hat{I}_B
		\right) = \gamma \hat{I}.
	\end{equation}
	Therefore, the local loss contribution becomes $-\frac{1}{2}\{\gamma \hat{I}, \rho(t)\} =
	-\gamma\rho(t)$. Replicating this exact derivation for Bob's independent local reset operators, the
	total master equation in Eq.~\eqref{eqn:ME:affine:noise} takes the compact, closed form:
	\begin{equation}
		\dot{\rho}(t) = \gamma ( \rho_c^A \otimes \mathrm{Tr}_A[\rho(t)] + \mathrm{Tr}_B[\rho(t)] \otimes
		\rho_c^B - 2\rho(t) ).
		\label{eqn:app:expanded_general_master}
	\end{equation}
	Now, consider the structural assumption that the initial correlated state $\rho_0$ has marginals
	matching those of the target product state, meaning $\mathrm{Tr}_A[\rho_0] = \rho_c^B$ and
	$\mathrm{Tr}_B[\rho_0] = \rho_c^A$. Taking the partial traces of
	Eq.~\eqref{eqn:app:expanded_general_master} reveals that these local marginals remain invariant
	under the dynamics ($\mathrm{Tr}_A[\rho(t)] = \rho_c^B$ and $\mathrm{Tr}_B[\rho(t)] = \rho_c^A$ for
	all $t \ge 0$). Crucially, this does not imply that the global state is invariant; while the local
	subsystems remain stationary, the global correlations are actively degraded. Under this invariant
	marginal condition, Eq.~\eqref{eqn:app:expanded_general_master} becomes:
	\begin{equation}
		\dot{\rho}(t) = 2\gamma \left( \rho_c^A \otimes \rho_c^B - \rho(t) \right) = -2\gamma \left(
		\rho(t) - \rho_c \right).
	\end{equation}
	Integrating this equation yields \erf{eqn:affine:evolution}.
	
	\section{Derivative of Robustness}  \label{appn:frechet:der}
	
	The robustness functional is defined as an optimization-based quantity over
	the convex classical set $\mathfrak C$. In order to study its continuous-time
	evolution, it is useful to interpret robustness as a functional over the
	space of density operators and introduce its differential structure.
	
	Let $\mathcal{B}(\mathcal H)$ denote the space of bounded operators acting on
	the Hilbert space $\mathcal H$. The Fr\'{e}chet derivative of a functional
	$F:\mathcal{B}(\mathcal H)\rightarrow\mathbb R$ at a point $\rho$ is the
	linear map $DF(\rho)$ satisfying~\cite{GiaHil04}
	\begin{equation} 	\label{eq:frechet:definition}
		F(\rho+\epsilon X) = F(\rho)+\epsilon\,DF(\rho)[X] + o(\epsilon),
	\end{equation}
	for an arbitrary perturbation operator $X\in\mathcal B(\mathcal H)$ and
	$\epsilon\rightarrow0$. Equivalently,
	\begin{equation} 	\label{eq:frechet:limit}
		DF(\rho)[X]=\lim_{\epsilon\rightarrow0}\frac{F(\rho+\epsilon X)-F(\rho)}{\epsilon},
	\end{equation}
	whenever this limit exists. For the robustness functional, this derivative is defined through by
	substituting $F \to R_{\mathfrak C}$.
	The operator $X$ represents an infinitesimal displacement in the quantum state space, and the
	Fr\'{e}chet derivative $DR_{\mathfrak C}(\rho)$ is a linear functional acting on perturbations of
	the state. Accordingly, $DR_{\mathfrak C}(\rho)[X]$ quantifies the corresponding
	first-order change of the robustness. In other words, $DR_{\mathfrak C}(\rho)[X]$ denotes the
	directional derivative of the robustness at the point $\rho$ in direction $X$.
	
	Unlike elementary functions, the robustness measure is obtained through a
	convex optimization problem, \erf{eqn:abs:R}.
	Consequently, differentiability of $R_{\mathfrak C}$ depends on the
	regularity of the optimal solution of this optimization problem. In regions
	where the optimal mixing parameter and optimal classical state vary
	smoothly with $\rho$, the robustness admits a well-defined Fr\'{e}chet
	derivative.
	
	At points where the optimal decomposition changes discontinuously, or where
	the state trajectory reaches a non-smooth portion of the boundary of
	$\mathfrak C$, the robustness functional may fail to be differentiable.
	In such cases, the appropriate mathematical description requires generalized
	derivatives or subdifferential calculus.
	
	\section{Werner state robustness} \label{appendix:werner}
	\subsection{Symmetry Reduction} \label{app:werner:symmetry}
	
	The discussion in Sec.~\ref{sec:iso:noise} uses the known Werner thresholds to calibrate the
	robustness framework. To justify this
	calibration, it remains to show that, for Werner states, the optimal reference state appearing in
	the absolute robustness optimization is indeed the maximally mixed state. The following argument
	establishes this result. Consider the two-qubit Werner family $\rho_{\rm W}(\wp)$ given in
	\erf{eqn:2q:werner}. The Werner family is invariant under collective unitary transformations of the
	form
	\begin{equation}
		(\hat U\otimes \hat U)\rho_{\rm W}(\hat U^\dagger\otimes \hat U^\dagger) =\rho_{\rm W},
	\end{equation}
	for all $\hat U\in \mathrm U(2)$. Associated with this symmetry is the twirling
	map~\cite{VolWer01}:
	\begin{equation}
		\mathcal T(\rho) = \int dU \, (\hat U\otimes \hat U) \,\rho \, (\hat U^\dagger\otimes \hat
		U^\dagger),
	\end{equation}
	where the integral is taken with respect to the Haar measure on $\mathrm U(2)$.
	The map $\mathcal T$ is completely positive and trace preserving, and satisfies
	\begin{equation}
		\mathcal T(\rho_{\rm W})=\rho_{\rm W},
		\qquad
		\mathcal T(\varrho_{\rm c})=\varrho_{\rm c}.
	\end{equation}
	Moreover, since $\mathcal T$ is a convex mixture of local unitary operations, it preserves the
	convex sets
	${\mathfrak C_{\mathrm m}}$. Suppose that for a given classical set $\mathfrak C_{\mathrm m}$ there
	exists a reference state $\sigma\in\mathfrak C_{\mathrm m}$ and parameter $\alpha\ge0$ such that
	\begin{equation}
		\frac{\rho_{\rm W}+\alpha\sigma}{1+\alpha} \in \mathfrak C_{\mathrm m}.
	\end{equation}
	Applying the twirling map and using linearity yields
	\begin{equation}
		\frac{\rho_{\rm W}+\alpha\,\mathcal T(\sigma)}{1+\alpha} \in \mathfrak C_{\mathrm m},
	\end{equation}
	where $\mathcal T(\sigma)\in\mathfrak C_{\mathrm m}$ by invariance of the classical set under
	twirling.
	Therefore every feasible point in the original optimization has a Werner-symmetric counterpart with
	the same admissible parameter $\alpha$. Since twirling cannot increase the minimum feasible value,
	the optimization defining $R_{\mathfrak C}$ may be restricted, without loss of generality, to
	Werner-symmetric reference states.
	Any two-qubit state invariant under collective twirling belongs to the Werner family.
	Hence the optimal reference state may be written as, $\sigma=\rho_{\rm W}(\wp)$.
	Consider now the convex mixture appearing in the robustness optimization:
	\begin{equation}
		\frac{\rho_{\rm W}(0)+\alpha\rho_{\rm W}(\wp)}{1+\alpha}.
	\end{equation}
	Using the linear parameterization of Werner states gives
	\begin{equation}
		\frac{\rho_{\rm W}(0)+\alpha \,\rho_{\rm W}(\wp)}{1+\alpha}=\rho_{\rm W}\left(\frac{\alpha
			\,\wp}{1+\alpha}\right).
	\end{equation}
	Hence, mixing two Werner states remains on the same isotropic line. Let $\wp_{\mathrm m}$ denote
	the
	critical Werner parameter at which the trajectory first enters the classical set $\mathfrak
	C_{\mathrm m}$.
	Then feasibility requires ${\alpha\,\wp}/({1+\alpha}) \ge \wp_{\mathrm m}$,
	which implies $\alpha \ge {\wp_{\mathrm m}}/({\wp-\wp_{\mathrm m}})$.
	For fixed threshold $\wp_{\mathrm m}$, the right-hand side is monotonically decreasing in $\wp$,
	and
	is minimized by the maximal admissible value $\wp=1$, corresponding precisely to the maximally
	mixed
	state
	\begin{equation}
		\sigma=\varrho_{\rm c}.
	\end{equation}
	Therefore, within the Werner-symmetric sector, the maximally mixed state constitutes an optimal
	reference state for the robustness optimization.
	It follows that, for the Werner family, the absolute robustness with respect to the classical sets
	$\mathfrak C_{\mathrm m}$ is realized by isotropic mixing with white noise. Consequently, in this
	context, the robustness optimization is exactly equivalent to determining the depolarizing
	threshold
	along the isotropic mixing direction. The robustness ordering therefore coincides with the
	experimentally accessible ordering of the depolarizing transition points.
	
	\subsection{Werner states under depolarizing noise}  \label{app:werner:depol}
	
	In this appendix we derive explicitly the evolution of two-qubit Werner states under a depolarizing
	channel and compute the corresponding transition times, \erf{eqn:w:trantime}, for the loss of Bell
	nonlocality, EPR steering, and entanglement within the robustness framework. We consider the family
	of Werner states $\rho_{\rm W}(\wp)$ given in \erf{eqn:2q:werner}. The noisy evolution is generated
	by a depolarizing channel acting along a fixed direction, $\mathcal E_t(\rho) =
	\left[1-\Lambda(t)\right]\rho + \Lambda(t)\varrho_c$,
	where $\Lambda(0)=0$ and $\Lambda(t)\to 1$ as $t\to\infty$. A standard choice is
	$\Lambda(t)=1-e^{-\gamma t}$ where $\gamma \in \mathbb R^+$ sets the timescale over which the
	depolarizing noise builds up. However, the derivation holds for any monotonic parameterization
	$\Lambda(t)$; the exponential form is introduced only to obtain the closed-form expression for the
	transition time. Applying the channel to $\rho_{\mathrm{W}}(\wp)$ yields the following state that
	still retains the Werner form:
	\begin{equation}
		\rho(t) = \left[1-\wp_{\mathrm{eff}}(t)\right]\,\rho_- + \wp_{\mathrm{eff}}(t)
		\varrho_{\mathrm{w}},
	\end{equation}
	with a time-dependent parameter $\wp_{\mathrm{eff}}(t) \equiv \Lambda(t) + 
	\left[1-\Lambda(t)\right]\wp$.
	Thus, the depolarizing channel acts by a simple rescaling of the Werner parameter, and the
	evolution
	remains within the Werner family.

	\subsection{Robustness decay rate along the Werner family} \label{app:werner:decay}
	
	The robustness decay rate introduced in the main text requires the
	functional dependence of the robustness measure along the one-dimensional
	Werner-state family. The symmetry reduction discussed above allows the absolute robustness
	optimization to be restricted to isotropic mixtures with the maximally
	mixed state. Consider an initial Werner state parameterized by
	$\wp_{\mathrm{eff}}$,
	\begin{equation}
		\rho_{\mathrm W}(\wp_{\mathrm{eff}}) = (1-\wp_{\mathrm{eff}}) \rho_- +
		\wp_{\mathrm{eff}}\varrho_{\mathrm c}.
	\end{equation}
	Let $\wp_{\mathfrak C}$ denote the boundary value at which the Werner
	family enters the classical set $\mathfrak C$. The robustness is obtained
	by finding the minimum amount of isotropic noise required to move the
	state to this boundary. Hence the optimal mixture satisfies
	\begin{equation}
		\frac{\rho_{\mathrm W}(\wp_{\mathrm{eff}})+\alpha \rho_{\mathrm W}(1)} {1+\alpha}	= \rho_{\mathrm
			W}(\wp_{\mathfrak C}).
	\end{equation}
	Replacing the Werner states in this expression and rearranging for $\wp_{\mathfrak C}$ leads to
	$({\wp_{\mathrm{eff}}+\alpha})/({1+\alpha})=\wp_{\mathfrak C}$. Solving for the optimal mixing
	parameter yields
	\begin{equation}
		R_{\mathfrak C}\left(\rho_{\mathrm W}(\wp_{\mathrm{eff}})\right)=\frac{\wp_{\mathfrak
				C}-\wp_{\mathrm{eff}}}{1-\wp_{\mathfrak C}},
	\end{equation}
	for $0\leq \wp_{\mathrm{eff}} \leq \wp_{\mathfrak C}$. Therefore the robustness varies linearly
	along the Werner direction. Substituting its derivative into the definition of the robustness decay
	rate, and assuming the exponential depolarizing model gives \erf{eqn:werner:decay:exp}.
	This result shows explicitly that the robustness decay consists of two
	factors: the channel-induced decay of the Werner parameter and the
	geometric sensitivity determined by the location of the classical
	boundary. 
	
	\section{Thresholds for Amplitude-Damped States}
	\label{app:asym:thresholds}
	
	In this appendix we explicitly determine the thresholds for Bell nonlocality and entanglement for
	the amplitude-damped singlet state considered in Sec.~\ref{sec:amp:damp}. The dynamics is generated
	by a local amplitude damping channel acting on Alice’s subsystem, $\rho(t) = (\mathcal{E}_A^{(t)}
	\otimes \mathcal{I}_B)(\rho)$,
	with Kraus operators given in \erfa{eqn:k0:amp:damp}{eqn:k1:amp:damp}, and damping parameter
	$\eta(t) = 1 - e^{-\gamma_{\rm a} t}$.
	Taking as initial state the singlet $\rho = |\psi_-\rangle\langle\psi_-|$, where $|\psi_-\rangle =
	(|01\rangle - |10\rangle)/\sqrt{2}$, in the basis $\{\ket{00},\ket{01},\ket{10},\ket{11}\}$ the
	evolved state can be written as
	\begin{equation}
		\rho(\eta) = \frac{1}{2}
		\begin{pmatrix}
			\eta & 0 & 0 & 0 \\
			0 & 1 & -\sqrt{1-\eta} & 0 \\
			0 &-\sqrt{1-\eta} & 1-\eta & 0 \\
			0 & 0 & 0 & 0
		\end{pmatrix}.
	\end{equation}
	This state is no longer invariant under exchange of the two subsystems, reflecting the asymmetric
	action of the noise.
	To determine the entanglement threshold, we employ the positive partial transpose (PPT) criterion.
	Taking the partial transpose with respect to Bob, the matrix $\rho^{\top_B}(\eta)$ becomes block
	diagonal. The only potentially negative contribution arises from the subspace associated with
	$\{|00\rangle, |11\rangle\}$. In this sector, the relevant eigenvalues are $\lambda_{\pm} =
	\tfrac{1}{2}\left[\eta \pm \sqrt{\eta^2 + 4(1-\eta)}\right]$.
	The smaller eigenvalue $\lambda_-$ is strictly negative for all $\eta < 1$ and vanishes only at
	$\eta = 1$. Therefore, the PPT criterion is violated for all finite damping strengths, and the
	state
	remains entangled for all $\eta < 1$. The entanglement threshold is thus $\eta_{\mathrm{ent}} = 1$.
	In terms of the time parameterization $\eta(t)$, this implies that entanglement persists for all
	finite times and is lost only asymptotically as $t \to \infty$.
	
	We next determine the threshold for Bell nonlocality using the Horodecki criterion
	\cite{Horodecki95}. For a two-qubit state $\rho$, define the correlation matrix $T \coloneq \langle
	\hat{\bm{\sigma}} \otimes \hat{\bm{\sigma}} \rangle$, where $\hat{\bm{\sigma}} \equiv (\hat
	\sigma_1, \hat \sigma_2, \hat \sigma_3)^\top$ is the Pauli operator basis, and $\langle \hat o
	\rangle = \mathrm{Tr}(\rho \hat o)$. The Clauser–Horne–Shimony–Holt (CHSH) inequality is violated
	if
	and only if
	\begin{equation}
		H(T) = s_1(T)^2 + s_2(T)^2 > 1,
	\end{equation}
	where $s_j(T)$ are the two largest singular values of $T$ (equivalently, the square roots of the
	eigenvalues of $T^\top T$).
	For the amplitude-damped singlet, the correlation matrix remains diagonal and takes the form
	$T(\eta) = \mathrm{diag}(-\sqrt{1-\eta},-\sqrt{1-\eta},\eta-1)$. It follows directly that
	$H(\eta)=2(1-\eta)$.
	The CHSH inequality is violated whenever $H(\eta) > 1$, yielding the Bell nonlocality threshold
	$\eta_{\mathrm{Bell}} = {1}/{2}$.
	In terms of the dynamical parameter, the corresponding transition time is $t_{\mathrm{Bell}} =
	{\gamma_{\rm a}^{-1}} \ln 2$.

	\bibliography{references}
	
\end{document}

%% file: author_list.tex
\author{Shakib Daryanoosh} 
\email{shakib.daryanoosh@curtin.edu.au}
\affiliation{Curtin Centre for Optimisation and Decision Science, Curtin University, Whadjuk Country, Perth 6102, Australia}

%
% visitor_addresses.tex                       26 March 2013 (a) 
%  available symbols are:
%  $\ast, \dag, \ddag, \S, \P, $\|$, $\ast\ast$, \dag\dag, \ddag\ddag ,\#
%

%{falkowski}
%$^{?}$Laboratoire de Physique Theorique, Orsay, FR,
%{hooper,kozminski}
%$^{?}$}Visitor from Lewis University, Romeoville, IL, USA.
%{weber}
%$^{?}$Universit{\"a}t Bern, Bern, Switzerland.
%{deceased}
%$^{\ddag}$Deceased.

%
\vskip 0.25cm

%% file: references.bib
@Preamble{"\providecommand{\noopsort}[1]{}" #
"\providecommand{\singleletter}[1]{#1}%"}

@article{WisDoh07,
	title = {Steering, Entanglement, Nonlocality, and the Einstein-Podolsky-Rosen Paradox},
	author = {Wiseman, H. M. and Jones, S. J. and Doherty, A. C.},
	journal = {Phys. Rev. Lett.},
	volume = {98},
	issue = {14},
	pages = {140402},
	numpages = {4},
	year = {2007},
	month = {Apr},
	publisher = {American Physical Society},
	doi = {10.1103/PhysRevLett.98.140402},
	url = {https://link.aps.org/doi/10.1103/PhysRevLett.98.140402}
}

@article{QuiBru15,
	title = {Inequivalence of entanglement, steering, and Bell nonlocality for general measurements},
	author = {Quintino, Marco T\'ulio and V\'ertesi, Tam\'as and Cavalcanti, Daniel and Augusiak, Remigiusz and Demianowicz, Maciej and Ac\'{\i}n, Antonio and Brunner, Nicolas},
	journal = {Phys. Rev. A},
	volume = {92},
	issue = {3},
	pages = {032107},
	numpages = {6},
	year = {2015},
	month = {Sep},
	publisher = {American Physical Society},
	doi = {10.1103/PhysRevA.92.032107},
	url = {https://link.aps.org/doi/10.1103/PhysRevA.92.032107}
}

@article{UolGuh14,
	title = {Joint Measurability of Generalized Measurements Implies Classicality},
	author = {Uola, Roope and Moroder, Tobias and G\"uhne, Otfried},
	journal = {Phys. Rev. Lett.},
	volume = {113},
	issue = {16},
	pages = {160403},
	numpages = {5},
	year = {2014},
	month = {Oct},
	publisher = {American Physical Society},
	doi = {10.1103/PhysRevLett.113.160403},
	url = {https://link.aps.org/doi/10.1103/PhysRevLett.113.160403}
}

@article{VidTar99,
	title = {Robustness of entanglement},
	author = {Vidal, Guifr\'e and Tarrach, Rolf},
	journal = {Phys. Rev. A},
	volume = {59},
	issue = {1},
	pages = {141--155},
	numpages = {0},
	year = {1999},
	month = {Jan},
	publisher = {American Physical Society},
	doi = {10.1103/PhysRevA.59.141},
	url = {https://link.aps.org/doi/10.1103/PhysRevA.59.141}
}

@article{HarNie03,
	title = {Robustness of quantum gates in the presence of noise},
	author = {Harrow, Aram W. and Nielsen, Michael A.},
	journal = {Phys. Rev. A},
	volume = {68},
	issue = {1},
	pages = {012308},
	numpages = {13},
	year = {2003},
	month = {Jul},
	publisher = {American Physical Society},
	doi = {10.1103/PhysRevA.68.012308},
	url = {https://link.aps.org/doi/10.1103/PhysRevA.68.012308}
}

@article{Werner89,
	title = {Quantum states with {E}instein-{P}odolsky-{R}osen correlations admitting a hidden-variable model},
	author = {Werner, Reinhard F.},
	journal = {Phys. Rev. A},
	volume = {40},
	issue = {8},
	pages = {4277--4281},
	numpages = {0},
	year = {1989},
	month = {Oct},
	publisher = {American Physical Society},
	doi = {10.1103/PhysRevA.40.4277},
	url = {https://link.aps.org/doi/10.1103/PhysRevA.40.4277}
}

@article{Horodecki96,
	title = {Separability of mixed states: necessary and sufficient conditions},
	journal = {Physics Letters A},
	volume = {223},
	number = {1},
	pages = {1-8},
	year = {1996},
	issn = {0375-9601},
	doi = {https://doi.org/10.1016/S0375-9601(96)00706-2},
	url = {https://www.sciencedirect.com/science/article/pii/S0375960196007062},
	author = {Micha{\l} Horodecki and Pawe{\l} Horodecki and Ryszard Horodecki}
}

@article{Horodecki95,
	title = {Violating Bell inequality by mixed spin-12 states: necessary and sufficient condition},
	journal = {Physics Letters A},
	volume = {200},
	number = {5},
	pages = {340-344},
	year = {1995},
	issn = {0375-9601},
	doi = {https://doi.org/10.1016/0375-9601(95)00214-N},
	url = {https://www.sciencedirect.com/science/article/pii/037596019500214N},
	author = {R. Horodecki and P. Horodecki and M. Horodecki}
}

@article{BruWeh14,
	title = {Bell nonlocality},
	author = {Brunner, Nicolas and Cavalcanti, Daniel and Pironio, Stefano and Scarani, Valerio and Wehner, Stephanie},
	journal = {Rev. Mod. Phys.},
	volume = {86},
	issue = {2},
	pages = {419--478},
	numpages = {60},
	year = {2014},
	month = {Apr},
	publisher = {American Physical Society},
	doi = {10.1103/RevModPhys.86.419},
	url = {https://link.aps.org/doi/10.1103/RevModPhys.86.419}
}

@article{UolGuh20,
	title = {Quantum steering},
	author = {Uola, Roope and Costa, Ana C. S. and Nguyen, H. Chau and G\"uhne, Otfried},
	journal = {Rev. Mod. Phys.},
	volume = {92},
	issue = {1},
	pages = {015001},
	numpages = {40},
	year = {2020},
	month = {Mar},
	publisher = {American Physical Society},
	doi = {10.1103/RevModPhys.92.015001},
	url = {https://link.aps.org/doi/10.1103/RevModPhys.92.015001}
}

@article{CavSkr16,
	title = {Quantitative relations between measurement incompatibility, quantum steering, and nonlocality},
	author = {Cavalcanti, D. and Skrzypczyk, P.},
	journal = {Phys. Rev. A},
	volume = {93},
	issue = {5},
	pages = {052112},
	numpages = {12},
	year = {2016},
	month = {May},
	publisher = {American Physical Society},
	doi = {10.1103/PhysRevA.93.052112},
	url = {https://link.aps.org/doi/10.1103/PhysRevA.93.052112}
}

@article{SkrCav14,
	title = {Quantifying {E}instein-{P}odolsky-{R}osen Steering},
	author = {Skrzypczyk, Paul and Navascu\'es, Miguel and Cavalcanti, Daniel},
	journal = {Phys. Rev. Lett.},
	volume = {112},
	issue = {18},
	pages = {180404},
	numpages = {5},
	year = {2014},
	month = {May},
	publisher = {American Physical Society},
	doi = {10.1103/PhysRevLett.112.180404},
	url = {https://link.aps.org/doi/10.1103/PhysRevLett.112.180404}
}

@article{PiaWat15,
	title = {Necessary and Sufficient Quantum Information Characterization of {E}instein-{P}odolsky-{R}osen Steering},
	author = {Piani, Marco and Watrous, John},
	journal = {Phys. Rev. Lett.},
	volume = {114},
	issue = {6},
	pages = {060404},
	numpages = {6},
	year = {2015},
	month = {Feb},
	publisher = {American Physical Society},
	doi = {10.1103/PhysRevLett.114.060404},
	url = {https://link.aps.org/doi/10.1103/PhysRevLett.114.060404}
}

@article{HorHor09,
	title = {Quantum entanglement},
	author = {Horodecki, Ryszard and Horodecki, Pawe\l{} and Horodecki, Micha\l{} and Horodecki, Karol},
	journal = {Rev. Mod. Phys.},
	volume = {81},
	issue = {2},
	pages = {865--942},
	numpages = {0},
	year = {2009},
	month = {Jun},
	publisher = {American Physical Society},
	doi = {10.1103/RevModPhys.81.865},
	url = {https://link.aps.org/doi/10.1103/RevModPhys.81.865}
}

@article{PirWal20,
	author = {S. Pirandola and U. L. Andersen and L. Banchi and M. Berta and D. Bunandar and R. Colbeck and D. Englund and T. Gehring and C. Lupo and C. Ottaviani and J. L. Pereira and M. Razavi and J. Shamsul Shaari and M. Tomamichel and V. C. Usenko and G. Vallone and P. Villoresi and P. Wallden},
	journal = {Adv. Opt. Photon.},
	number = {4},
	pages = {1012--1236},
	publisher = {Optica Publishing Group},
	title = {Advances in quantum cryptography},
	volume = {12},
	month = {Dec},
	year = {2020},
	url = {https://opg.optica.org/aop/abstract.cfm?URI=aop-12-4-1012},
	doi = {10.1364/AOP.361502}
}

@article{MonBay25,
	title = {Review: Quantum metrology and sensing with many-body systems},
	journal = {Physics Reports},
	volume = {1134},
	pages = {1-62},
	year = {2025},
	issn = {0370-1573},
	doi = {https://doi.org/10.1016/j.physrep.2025.05.005},
	url = {https://www.sciencedirect.com/science/article/pii/S0370157325001565},
	author = {Victor Montenegro and Chiranjib Mukhopadhyay and Rozhin Yousefjani and Saubhik Sarkar and Utkarsh Mishra and Matteo G.A. Paris and Abolfazl Bayat}
}

@article{BhaAsp22,
	title = {Noisy intermediate-scale quantum algorithms},
	author = {Bharti, Kishor and Cervera-Lierta, Alba and Kyaw, Thi Ha and Haug, Tobias and Alperin-Lea, Sumner and Anand, Abhinav and Degroote, Matthias and Heimonen, Hermanni and Kottmann, Jakob S. and Menke, Tim and Mok, Wai-Keong and Sim, Sukin and Kwek, Leong-Chuan and Aspuru-Guzik, Al\'an},
	journal = {Rev. Mod. Phys.},
	volume = {94},
	issue = {1},
	pages = {015004},
	numpages = {69},
	year = {2022},
	month = {Feb},
	publisher = {American Physical Society},
	doi = {10.1103/RevModPhys.94.015004},
	url = {https://link.aps.org/doi/10.1103/RevModPhys.94.015004}
}

@article{Zur03,
	title = {Decoherence, einselection, and the quantum origins of the classical},
	author = {Zurek, Wojciech Hubert},
	journal = {Rev. Mod. Phys.},
	volume = {75},
	issue = {3},
	pages = {715--775},
	numpages = {0},
	year = {2003},
	month = {May},
	publisher = {American Physical Society},
	doi = {10.1103/RevModPhys.75.715},
	url = {https://link.aps.org/doi/10.1103/RevModPhys.75.715}
}

@article{TinEbe09,
	author = {Ting Yu  and J. H. Eberly },
	title = {Sudden Death of Entanglement},
	journal = {Science},
	volume = {323},
	number = {5914},
	pages = {598-601},
	year = {2009},
	doi = {10.1126/science.1167343},
	URL = {https://www.science.org/doi/abs/10.1126/science.1167343}
}

@article{DarSlu24,
	title = {Experimental investigation of a multiphoton Heisenberg-limited interferometric scheme: The effect of imperfections},
	author = {Daryanoosh, Shakib and Pryde, Geoff J. and Wiseman, Howard M. and Slussarenko, Sergei},
	journal = {Phys. Rev. A},
	volume = {110},
	issue = {1},
	pages = {012614},
	numpages = {16},
	year = {2024},
	month = {Jul},
	publisher = {American Physical Society},
	doi = {10.1103/PhysRevA.110.012614},
	url = {https://link.aps.org/doi/10.1103/PhysRevA.110.012614}
}

@article{AciWil18,
	doi = {10.1088/1367-2630/aad1ea},
	url = {https://doi.org/10.1088/1367-2630/aad1ea},
	year = {2018},
	month = {aug},
	publisher = {IOP Publishing},
	volume = {20},
	number = {8},
	pages = {080201},
	author = {Acín, Antonio and Bloch, Immanuel and Buhrman, Harry and Calarco, Tommaso and Eichler, Christopher and Eisert, Jens and Esteve, Daniel and Gisin, Nicolas and Glaser, Steffen J and Jelezko, Fedor and Kuhr, Stefan and Lewenstein, Maciej and Riedel, Max F and Schmidt, Piet O and Thew, Rob and Wallraff, Andreas and Walmsley, Ian and Wilhelm, Frank K},
	title = {The quantum technologies roadmap: a European community view},
	journal = {New Journal of Physics}
}

@article{StrHer11,
	title = {Behavior of Quantum Correlations under Local Noise},
	author = {Streltsov, Alexander and Kampermann, Hermann and Bru\ss{}, Dagmar},
	journal = {Phys. Rev. Lett.},
	volume = {107},
	issue = {17},
	pages = {170502},
	numpages = {5},
	year = {2011},
	month = {Oct},
	publisher = {American Physical Society},
	doi = {10.1103/PhysRevLett.107.170502},
	url = {https://link.aps.org/doi/10.1103/PhysRevLett.107.170502}
}

@article{IkeLow24,
	title = {Robustness of quantum correlation in quantum energy teleportation},
	author = {Ikeda, Kazuki and Lowe, Adam},
	journal = {Phys. Rev. D},
	volume = {110},
	issue = {9},
	pages = {096010},
	numpages = {11},
	year = {2024},
	month = {Nov},
	publisher = {American Physical Society},
	doi = {10.1103/PhysRevD.110.096010},
	url = {https://link.aps.org/doi/10.1103/PhysRevD.110.096010}
}

@article{Bell64,
	title = {On the Einstein Podolsky Rosen paradox},
	author = {Bell, J. S.},
	journal = {Physics Physique Fizika},
	volume = {1},
	issue = {3},
	pages = {195--200},
	numpages = {6},
	year = {1964},
	month = {Nov},
	publisher = {American Physical Society},
	doi = {10.1103/PhysicsPhysiqueFizika.1.195},
	url = {https://link.aps.org/doi/10.1103/PhysicsPhysiqueFizika.1.195}
}

@Article{HenHan15,
	author={Hensen, B.
	and Bernien, H.
	and Dr{\'e}au, A. E.
	and Reiserer, A.
	and Kalb, N.
	and Blok, M. S.
	and Ruitenberg, J.
	and Vermeulen, R. F. L.
	and Schouten, R. N.
	and Abell{\'a}n, C.
	and Amaya, W.
	and Pruneri, V.
	and Mitchell, M. W.
	and Markham, M.
	and Twitchen, D. J.
	and Elkouss, D.
	and Wehner, S.
	and Taminiau, T. H.
	and Hanson, R.},
	title={Loophole-free Bell inequality violation using electron spins separated by 1.3 kilometres},
	journal={Nature},
	year={2015},
	month={Oct},
	day={01},
	volume={526},
	number={7575},
	pages={682-686},
	issn={1476-4687},
	doi={10.1038/nature15759},
	url={https://doi.org/10.1038/nature15759}
}

@article{GiuZei15,
	title = {Significant-Loophole-Free Test of Bell's Theorem with Entangled Photons},
	author = {Giustina, Marissa and Versteegh, Marijn A. M. and Wengerowsky, S\"oren and Handsteiner, Johannes and Hochrainer, Armin and Phelan, Kevin and Steinlechner, Fabian and Kofler, Johannes and Larsson, Jan-\AA{}ke and Abell\'an, Carlos and Amaya, Waldimar and Pruneri, Valerio and Mitchell, Morgan W. and Beyer, J\"orn and Gerrits, Thomas and Lita, Adriana E. and Shalm, Lynden K. and Nam, Sae Woo and Scheidl, Thomas and Ursin, Rupert and Wittmann, Bernhard and Zeilinger, Anton},
	journal = {Phys. Rev. Lett.},
	volume = {115},
	issue = {25},
	pages = {250401},
	numpages = {7},
	year = {2015},
	month = {Dec},
	publisher = {American Physical Society},
	doi = {10.1103/PhysRevLett.115.250401},
	url = {https://link.aps.org/doi/10.1103/PhysRevLett.115.250401}
}

@article{ShaNam15,
	title = {Strong Loophole-Free Test of Local Realism},
	author = {Shalm, Lynden K. and Meyer-Scott, Evan and Christensen, Bradley G. and Bierhorst, Peter and Wayne, Michael A. and Stevens, Martin J. and Gerrits, Thomas and Glancy, Scott and Hamel, Deny R. and Allman, Michael S. and Coakley, Kevin J. and Dyer, Shellee D. and Hodge, Carson and Lita, Adriana E. and Verma, Varun B. and Lambrocco, Camilla and Tortorici, Edward and Migdall, Alan L. and Zhang, Yanbao and Kumor, Daniel R. and Farr, William H. and Marsili, Francesco and Shaw, Matthew D. and Stern, Jeffrey A. and Abell\'an, Carlos and Amaya, Waldimar and Pruneri, Valerio and Jennewein, Thomas and Mitchell, Morgan W. and Kwiat, Paul G. and Bienfang, Joshua C. and Mirin, Richard P. and Knill, Emanuel and Nam, Sae Woo},
	journal = {Phys. Rev. Lett.},
	volume = {115},
	issue = {25},
	pages = {250402},
	numpages = {10},
	year = {2015},
	month = {Dec},
	publisher = {American Physical Society},
	doi = {10.1103/PhysRevLett.115.250402},
	url = {https://link.aps.org/doi/10.1103/PhysRevLett.115.250402}
}

@article{Steiner03,
	title = {Generalized robustness of entanglement},
	author = {Steiner, Michael},
	journal = {Phys. Rev. A},
	volume = {67},
	issue = {5},
	pages = {054305},
	numpages = {4},
	year = {2003},
	month = {May},
	publisher = {American Physical Society},
	doi = {10.1103/PhysRevA.67.054305},
	url = {https://link.aps.org/doi/10.1103/PhysRevA.67.054305}
}

@article{TakReg19,
	title = {General Resource Theories in Quantum Mechanics and Beyond: Operational Characterization via Discrimination Tasks},
	author = {Takagi, Ryuji and Regula, Bartosz},
	journal = {Phys. Rev. X},
	volume = {9},
	issue = {3},
	pages = {031053},
	numpages = {28},
	year = {2019},
	month = {Sep},
	publisher = {American Physical Society},
	doi = {10.1103/PhysRevX.9.031053},
	url = {https://link.aps.org/doi/10.1103/PhysRevX.9.031053}
}

@article{HorRus03,
	author = {Horodecki, Michael and Shor, Peter W. and Ruskai, Mary Beth},
	title = {Entanglement Breaking Channels},
	journal = {Reviews in Mathematical Physics},
	volume = {15},
	number = {06},
	pages = {629-641},
	year = {2003},
	doi = {10.1142/S0129055X03001709},
	URL = {https://doi.org/10.1142/S0129055X03001709}
}

@article{GouSca21,
	title = {Entanglement of a bipartite channel},
	author = {Gour, Gilad and Scandolo, Carlo Maria},
	journal = {Phys. Rev. A},
	volume = {103},
	issue = {6},
	pages = {062422},
	numpages = {21},
	year = {2021},
	month = {Jun},
	publisher = {American Physical Society},
	doi = {10.1103/PhysRevA.103.062422},
	url = {https://link.aps.org/doi/10.1103/PhysRevA.103.062422}
}

@book{NieChu10, 
	place={Cambridge}, 
	title={Quantum Computation and Quantum Information: 10th Anniversary Edition}, 
	publisher={Cambridge University Press}, 
	author={Nielsen, Michael A. and Chuang, Isaac L.}, 
	year={2010}
}

@book{GiaHil04,
	author    = {Mariano Giaquinta and Stefan Hildebrandt},
	title     = {Calculus of Variations I},
	series    = {Grundlehren der mathematischen Wissenschaften},
	volume    = {310},
	publisher = {Springer},
	address   = {Berlin, Heidelberg},
	year      = {2004},
	doi       = {10.1007/978-3-662-03278-7},
	isbn      = {978-3-662-03276-3},
	isbn      = {978-3-662-03278-7}
}

@book{Ste08,
	author    = {Steven Roman},
	title     = {Advanced Linear Algebra},
	edition   = {3},
	series    = {Graduate Texts in Mathematics},
	volume    = {135},
	publisher = {Springer},
	address   = {New York, NY},
	year      = {2008},
	isbn      = {978-0-387-72828-5},
	doi       = {10.1007/978-0-387-72831-5},
	url       = {https://doi.org/10.1007/978-0-387-72831-5}
}

@article{GKS76,
	author = {Gorini, Vittorio and Kossakowski, Andrzej and Sudarshan, E. C. G.},
	title = {Completely positive dynamical semigroups of N‐level systems},
	journal = {Journal of Mathematical Physics},
	volume = {17},
	number = {5},
	pages = {821-825},
	year = {1976},
	month = {05},
	issn = {0022-2488},
	doi = {10.1063/1.522979},
	url = {https://doi.org/10.1063/1.522979}
}

@Article{Lindblad76,
	author={Lindblad, G.},
	title={On the generators of quantum dynamical semigroups},
	journal={Communications in Mathematical Physics},
	year={1976},
	month={Jun},
	day={01},
	volume={48},
	number={2},
	pages={119-130},
	issn={1432-0916},
	doi={10.1007/BF01608499},
	url={https://doi.org/10.1007/BF01608499}
}

@article{KuNor18,
	title = {Hierarchy in temporal quantum correlations},
	author = {Ku, Huan-Yu and Chen, Shin-Liang and Lambert, Neill and Chen, Yueh-Nan and Nori, Franco},
	journal = {Phys. Rev. A},
	volume = {98},
	issue = {2},
	pages = {022104},
	numpages = {9},
	year = {2018},
	month = {Aug},
	publisher = {American Physical Society},
	doi = {10.1103/PhysRevA.98.022104},
	url = {https://link.aps.org/doi/10.1103/PhysRevA.98.022104}
}

@article{VolWer01,
	title = {Entanglement measures under symmetry},
	author = {Vollbrecht, K. G. H. and Werner, R. F.},
	journal = {Phys. Rev. A},
	volume = {64},
	issue = {6},
	pages = {062307},
	numpages = {15},
	year = {2001},
	month = {Nov},
	publisher = {American Physical Society},
	doi = {10.1103/PhysRevA.64.062307},
	url = {https://link.aps.org/doi/10.1103/PhysRevA.64.062307}
}
